\UseRawInputEncoding
\documentclass[letterpaper, 10 pt, conference]{ieeeconf}  

\IEEEoverridecommandlockouts                              

\usepackage{graphics} 
\usepackage{epsfig} 
\usepackage{amsmath} 
\usepackage{amssymb}  
\usepackage{subcaption} 
\usepackage{color,soul}
\usepackage{mathtools}

\newtheorem{thm}{Theorem}

\newtheorem{cor}{Corollary}

\newtheorem{defn}{Definition}

\newtheorem{rem}{Remark}

\title{\LARGE \bf
Region of Attraction Estimate Learning and Verification for Nonlinear Systems using Neural-Network-based Lyapunov Functions
}

\author{Adel Bechihi$^{1}$, \textit{Member, IEEE}, and Aristotelis Kapnopoulos$^{1}$
\thanks{$^{1}$ Adel Bechihi and Aristotelis Kapnopoulos are with Mitsubishi Electric R\&D Centre
	Europe, 35700, Rennes, France. Email:
        {\tt\small a.bechihi@fr.merce.mee.com} and {\tt\small a.kapnopoulos@fr.merce.mee.com}.}%
}

\begin{document}

\maketitle
\thispagestyle{empty}
\pagestyle{empty}

\begin{abstract}
	Estimating the Region of Attraction (RoA) for nonlinear dynamical systems is a fundamental problem in control theory, with direct implications for stability analysis and safe controller design. Traditional approaches rely on analytically derived Lyapunov functions, which are often conservative and challenging to construct for high-dimensional or highly nonlinear systems. In this work, we propose a data-driven framework for learning and verifying RoA estimates for nonlinear systems using neural-network-based Lyapunov functions.	
	Our method employs a composite Lyapunov function that combines a quadratic term with a neural-network-based component, providing both structure and flexibility. We introduce a novel homogeneous loss function for training, which removes the imbalance typically caused by the two non-homogeneous Lyapunov conditions. Together, these two aspects enable efficient training of the Lyapunov candidate. To guarantee the correctness of the learned Lyapunov function, we employ a Satisfiability Modulo Theories (SMT) solver to formally verify the stability results. Lastly, we perform a deeper analysis near the origin to overcome numerical artifacts, ensuring strict asymptotic stability.	
	We demonstrate the effectiveness of our approach on benchmark nonlinear systems, showing that it significantly reduces conservatism compared to traditional Lyapunov methods while maintaining verifiability. This framework bridges the gap between function approximation and stability certification, paving the way for scalable safety analysis in learning-based control and safety-critical applications.
\end{abstract}

\section{Introduction}

Learning-based approaches have recently become a cornerstone for tackling complex nonlinear control problems, with neural networks playing a central role in this transformation. Their ability to approximate highly nonlinear mappings has enabled the development of control policies through interaction with the environment, achieving impressive results across a wide range of robotic platforms \cite{c1,c2}. Yet, despite these successes, a fundamental limitation persists: most learning-based controllers lack formal stability guarantees. This absence of theoretical assurance raises critical concerns about safety and interpretability, particularly in high-stakes applications where reliability is paramount. Bridging this gap between empirical performance and rigorous certification remains an open challenge.

Lyapunov theory provides a foundational tool for stability analysis, known as Lyapunov functions, offering a systematic way to certify system stability without explicitly solving differential equations. Moreover, Lyapunov functions can be used in a straightforward way to infer an estimate of the region of attraction (RoA) of the studied system \cite{c6}. This motivated numerous approaches for constructing Lyapunov functions across different classes of dynamical systems. For linear systems, Lyapunov functions can be derived by solving algebraic equations \cite{c6}, while for polynomial systems, sum-of-squares (SoS) techniques enable Lyapunov synthesis through semi-definite programs (SDPs) \cite{c7}. Despite their strong theoretical foundations, these approaches impose restrictive structural assumptions that limit their applicability to general nonlinear systems. Furthermore, SDP-based methods suffer from scalability and numerical sensitivity issues \cite{c8}, making them computationally infeasible for high-dimensional or high-degree systems.

To overcome these limitations, recent research has increasingly explored learning-based methods for Lyapunov function synthesis. Neural networks, supported by the universal approximation theorem \cite{c11}, can approximate any continuous function on a compact domain, making them natural candidates for representing complex Lyapunov functions for nonlinear systems. Early works demonstrated the feasibility of this idea by approximating Lyapunov functions with shallow networks \cite{c13}, marking one of the first attempts to apply neural approximators in stability analysis. Subsequent studies proposed architectures and training strategies to jointly learn Lyapunov functions and stabilizing controllers. For instance, Chang et al. \cite{c16} investigated Lyapunov learning with quadratic and $\tanh$ activation functions, while Dai et al. \cite{c14} introduced a mixed-integer programming (MIP) approach to verify stability conditions during training. More recent works have proposed specialized architectures to enforce positive definiteness and other Lyapunov properties by design. In \cite{c15} authors propose a specialized network architecture to ensure positive definiteness in Lyapunov functions. However, this requires activation functions with trivial null space, which excludes many common choices like ReLU, sigmoid, and tanh functions. Additionally, their approach can cause the network output to grow excessively fast as the input norm increases and conservative stability conditions due to reliance on Lipschitz bound limit scalability and practical applicability. In \cite{c17} a Lyapunov neural network  architecture is proposed for high dimensional systems that offers theoretical guarantees for approximating Lyapunov functions with bounded Lipschitz constants, aiding robust certification. However, this can lead to a conservative and costly verification, certifying only a delta-accurate Lyapunov function, guaranteeing asymptotic stability only outside a neighborhood around the origin due to the non-strict negativity near the origin of the Lyapunov function derivative.

Synthesizing neural Lyapunov functions requires verifying that the trained function satisfies Lyapunov conditions across a specified domain, since learning-based approaches rely data that represent only samples with a dense domain representing the physical space. Existing verification techniques fall into three categories, with the first one relying on generalization bounds, computationally efficient but offering only probabilistic guarantees. The second relies on Lipschitz-based analysis to provide deterministic guarantees, though at higher computational cost. The third, optimization-based methods, bridge this gap by verifying the certificate throughout training via a learner-verifier loop also known as counterexample guided inductive synthesis, often using satisfiability modulo theories (SMT) \cite{c16} or MIP \cite{c14}, solvers to find counterexamples and refine the certificate.

In this work, we present a framework for learning Lyapunov functions that addresses these challenges comprehensively. Our method guarantees positive definiteness by design, making the training process easier and faster. We also introduce a new loss function, carefully designed to allow the increase of the RoA estimate size during training. To provide soundness of the solution, we employ an SMT-based verification method to identify states that violate the Lyapunov conditions. The SMT solver used in this work is dReal \cite{c20}. Moreover, our method guarantees that the trained Lyapunov neural network satisfies the Lyapunov conditions for all states within the verified domain, bypassing the common limitation of delta-Lyapunov functions \cite{c16,c17}, which fail to certify asymptotic stability near the origin. We validate the proposed framework on two nonlinear dynamical systems examples, demonstrating its ability to synthesize Lyapunov functions and infer RoA estimates efficiently. 

The paper is organized as follows. First, we present preliminary results and formulate the problem. Next, we describe our methodology for constructing RoA estimates, including design, training, and verification strategies. Finally, we provide two numerical examples and conclude the paper.

\section{Preliminaries and problem formulation}
We consider an n-dimensional nonlinear continuous-time dynamical system described by the differential equation:
\begin{align}
	\label{eq:system_dynamics}
	\dot{x} = f(x),
\end{align}
where $x \in \mathcal{D} \subset \mathbb{R}^n$ is the system state and $f: \mathcal{D} \to \mathbb{R}^n$ is the flow map which we assume to be locally Lipschitz in $x$. 

\begin{thm}[Lyapunov's indirect method \cite{c6}]
	\label{thm:lyap_indirect}
	Let $x_e$ be the equilibrium point for the dynamical system \eqref{eq:system_dynamics}. Let $A = \frac{\partial f}{\partial x}\big|_{x=x_e}$ be the Jacobian matrix of $f$ at $x_e$. Then, $x_e$ is asymptotically stable if all the eigenvalues of $A$ have a strictly negative real part.
\end{thm}

Theorem \ref{thm:lyap_indirect} provides simple tool to check the stability properties of a given equilibrium point providing local stability information around it. However, it does not quantify how far this stability extends.

\begin{defn}[Region of attraction \cite{c6}]
	Let $\phi(t;x,t_0)$ be the solution of the dynamical system \eqref{eq:system_dynamics} at time $t$ that starts at initial state $x \in \mathcal{D}$ at time initial time $t_0$. Suppose that $\phi(t;x,t_0)$ is defined for all $t \ge t_0$. The region of attraction $\mathcal{R}$ of the equilibrium point $x_e = 0 $ is defined by:
	\begin{align*}
		\mathcal{R} = \{ x \in \mathcal{D} | \lim_{t \to \infty} \phi(t;x,t_0) = 0\}. 
	\end{align*} 	
\end{defn}

Estimating this region is crucial in nonlinear systems analysis, as it characterizes the robustness of stability and informs safe operating conditions. Unlike linear systems, where stability implies global convergence, nonlinear systems often exhibit limited regions of attraction, making their estimation a central problem in control theory.
 
\begin{thm}[Lyapunov function \cite{c6}]
	\label{def:LF}
	Let $x_e = 0$ be an equilibrium point for the dynamical system \eqref{eq:system_dynamics} and $\mathcal{D} \subset \mathbb{R}^n$ be a domain containing $x_e$. Let $V: \mathcal{D} \to \mathbb{R}$ be a continuously differentiable function such that
	\begin{align}
		&V(0) = 0 \text{ and } V(x) > 0;  \forall x \in \mathcal{D} \setminus \{0\}, \label{eq:positive_definite_V} \\
		&\dot{V}(x) = \nabla V(x)^\top f(x) < 0;  \forall x \in \mathcal{D} \setminus \{0\},  \label{eq:negative_Vdot}
	\end{align}
	then $x_e = 0$ is asymptotically stable.
\end{thm}

\begin{cor}[RoA estimate \cite{c6}]
	\label{cor:RoA_estimate}
	Let $V: \mathcal{D} \to \mathbb{R}$ be a Lyapunov function for the dynamical system \eqref{eq:system_dynamics}. Every compact sub-level set $\Omega_c = \{ x \in \mathcal{D} | V(x) \le c\}$ that satisfies $\Omega_c \subset \mathcal{H} = \{x \in \mathcal{D} | \dot{V}(x) < 0 \} \cup \{0\}$ is an inner estimate of the region of attraction with $\Omega_c \subset \mathcal{R} \subset \mathcal{D}$.
\end{cor}

Corollary \ref{cor:RoA_estimate} provides a constructive approach for estimating the RoA. However, the size of the resulting estimate is strongly influenced by the choice of the Lyapunov function $V$, defined over the domain $\mathcal{D}$, and by the level set parameter $c$. To obtain the largest possible RoA estimate, our objective is to develop a systematic method that jointly determines an appropriate Lyapunov function $V$ and its corresponding maximal admissible level $c_{\max}$, defined as:
\begin{align}
	\label{eq:c_max}
	c_{max} = \sup\left\{ c>0 | \forall x \in \Omega_c \setminus \{0\}, \dot{V}(x) < 0 \right\},
\end{align}
where $\Omega_c = \{ x \in \mathcal{D} | V(x) \le c\}$. This formulation ensures that the level set $\Omega_{c_{\max}}$ is forward invariant and provides the largest certified inner approximation of the true RoA.

Directly searching over the space of all admissible Lyapunov functions is an infinite-dimensional problem, which is computationally intractable. To overcome this challenge, we restrict our search to a parametric family of Lyapunov functions $V_\theta$, where $\theta \in \mathbb{R}^p$ denotes a finite-dimensional parameter vector. This parametrization transforms the original infinite-dimensional optimization into a finite-dimensional one, making it amenable to numerical methods while preserving flexibility in the choice of $V$.

\section{Methodology}
In this section we present our method to design a parametric Lyapunov function for the dynamical system \eqref{eq:system_dynamics}, and infer its corresponding largest RoA estimate. The proposed function $V_\theta:\mathbb{R}^n \to \mathbb{R}$ is optimized to maximize the size of the RoA estimate through a training process. Finally, a verification process is presented to certify the correctness to the deduced RoA estimate. 



\subsection{Design of Lyapunov-based loss function}
\subsubsection{Lyapunov candidate function design}
\label{subsec:design_LF}

We propose a Lyapunov candidate function that incorporates a feed-forward neural network, which is trained to approximate functions satisfying Lyapunov conditions for assessing the stability of nonlinear dynamical systems. The network comprises an input layer with $n$ neurons, corresponding to the dimension of the state, a single hidden layer with $h_1$ neurons, and an output layer with $m$ neurons. Although the general formulation allows for an arbitrary number of hidden layers $k \geq 0$, the present analysis is restricted to the case $k=1$, which facilitates theoretical tractability and interpretability.

Each neuron in the hidden layer applies a nonlinear activation function, which is critical for enabling the approximation of  complex, non-convex functions. The choice of activation function significantly influences the smoothness, differentiability, and boundedness of the Lyapunov candidate function. 

The general one-hidden-layer feed-forward neural network architecture is expressed by the following equations:
\begin{align*}
	z &= \sigma(W_1 x + b_1), \\
	y &= W_2 z + b_2,
\end{align*}
where $x \in \mathbb{R}^n$ is the input vector representing the system state, $W_1, b_1, W_2, b_2$ are the corresponding weights and biases of the hidden and output layers respectively, and $y \in \mathbb{R}^m$ is the neural network output. The function $\sigma:\mathbb{R}\to\mathbb{R}$ is a given nonlinear activation function that is applied element-wise to each node of the hidden layer. 

We write the previous equations of the neural network in a compact form as $y = \varphi_\theta(x)$, where $\theta = \{W_1,b_1,W_2,b_2\}$ denotes the set of all trainable parameters of the network, and $\varphi_\theta:\mathbb{R}^n\to\mathbb{R}^m$ represents the neural network mapping.

Thus, we propose the Following neural-network-based Lyapunov candidate function:
\begin{align}
	\label{eq:our_lyapunov_function}
	V_\theta(x) = \| \varphi_\theta(x) - \varphi_\theta(x_e) \|^2 + (x-x_e)^TP_\theta (x-x_e),
\end{align}
where $P_\theta \in \mathbb{R}^{n \times n}$ is a symmetric and positive definite matrix, whose construction is detailed in subsection \ref{sec:analysis_origin}.

It is easy to verify that $V_\theta$ is a positive definite function by design, thus satisfying the condition \eqref{eq:positive_definite_V}, for any neural network structure $\varphi_\theta$ and any parameters $\theta$. In the sequel, we restrict our analysis to neural network architectures with a scalar output, i.e., $m = 1$, and select the $\tanh$ function as the activation function due to its smoothness properties. Note that the bias term $b_2$​ in the output layer is unnecessary in this formulation, since the Lyapunov candidate function $V_\theta$ depends only on $\varphi_\theta(x) - \varphi_\theta(x_e)$, which cancels out any constant offset. Therefore, the set of all trainable parameters is restricted to $\theta = \{W_1,b_1,W_2\}$. Without loss of generality, we assume that the origin $x_e = 0$ is a locally asymptotically stable equilibrium point of system \eqref{eq:system_dynamics} for the remainder of this paper, and restrict our analysis to it.

\subsubsection{Loss function design for maximizing the size of the region of attraction estimate}

To verify the Lyapunov conditions given in Theorem \ref{def:LF} for the proposed  Lyapunov candidate function \eqref{eq:our_lyapunov_function}, we design a learning strategy to adapt the network parameters $\theta$ such that resulting Lyapunov function allows to infer a RoA estimate.

We first state the following theorem.
\begin{thm}
	\label{our_theorem}
	Let $V_\theta:\mathcal{D}\to\mathbb{R}$ be continuously differentiable on the domain $\mathcal{D}\subset\mathbb{R}^n$ containing the origin, with $V_\theta(0)=0$ and $V_\theta(x)>0$ for all $x\neq 0$. Consider the system \eqref{eq:system_dynamics} and define 
	\begin{align*}
		&\hat{\mathcal{R}} = \nonumber \\ 
		&\{x \in \mathcal{D} | \forall y\in \mathcal{D} \setminus \{0\}, V_\theta(y) \le V_\theta(x) \Rightarrow \dot{V_\theta}(y)<0\} \cup \{0\},
	\end{align*}
	where $\dot V_\theta(y) = \nabla V_\theta(y)^T f(y)$. Then $\hat{\mathcal{R}}$ is positively invariant and is contained in the region of attraction of the origin. Moreover, $\hat{\mathcal{R}}$ coincides with $\Omega_{c_{\max}}$ introduced in \eqref{eq:c_max}.
\end{thm}

\begin{proof}
	Let $x \in \hat{\mathcal{R}} \setminus \{0\}$ and set $c_0 := V_\theta(x)$. By definition of $\hat{\mathcal{R}}$, for every $y \in \mathcal{D}$ with $V_\theta(y) \le c_0$ and $y \neq 0$, we have $\dot V_\theta(y) < 0$. Define the sublevel set	$\Omega_{c_0} := \{ y \in \mathcal{D} | V_\theta(y) \le c_0 \}$.	On $\Omega_{c_0} \setminus \{0\}$, $\dot V_\theta < 0$ holds everywhere. Therefore, along any trajectory starting in $\Omega_{c_0}$, the function $t \mapsto V_\theta(x(t))$ is strictly decreasing until it reaches its minimum at $0$. This implies that $\Omega_{c_0}$ is forward invariant because $V_\theta(x(t))$ cannot exceed its initial value $c_0$. Moreover, every trajectory starting in $\Omega_{c_0}$ converges to the origin by LaSalle's invariance principle, since the largest invariant subset of $\{ y \in \Omega_{c_0} | \dot V_\theta(y) = 0 \}$ is $\{0\}$.
	
	To prove that  $\hat{\mathcal{R}}$ coincides with $\Omega_{c_{\max}}$ introduced in \eqref{eq:c_max}, consider $x \in \Omega_{c_{\max}}$ and set $c = V_\theta(x)$. By definition of $c_{\max}$, all $y \in \mathcal{D}$ with $V_\theta(y) \le c$ satisfy $\dot V_\theta(y) < 0$ except at $y=0$. Therefore $x \in \hat{\mathcal{R}}$. Conversely, if $x \in \hat{\mathcal{R}}$, then the strict negativity condition holds on $\Omega_{V_\theta(x)}$, so $V_\theta(x) \le c_{\max}$ by maximality of $c_{\max}$. Hence $\hat{\mathcal{R}} = \Omega_{c_{\max}}$.
\end{proof}

The above definition for $\hat{\mathcal{R}}$ allow us to characterize the largest compact invariant set, for a given Lyapunov function $V_\theta$ without the need of explicitly computing the highest level $c_{max}$. While Theorem \ref{our_theorem} provides a way to identify all points in the continuous  domain $\mathcal{D}$ that lie within the RoA estimate $\hat{\mathcal{R}}$, in practice we work with a sampled approximation. To do this, we shift our attention to a grid $\mathcal{G} \subset \mathcal{D}$, which serves as a discretization of the state space. Our goal is to approximate the RoA estimate described by Theorem \ref{our_theorem} using the grid $\mathcal{G}$. To this end, we introduce the finite countable set $\hat{\mathcal{R}}_\mathcal{G} = \hat{\mathcal{R}} \cap \mathcal{G}$, with $\mathcal{G} = \{x_i \in \mathcal{D} | i=1,...,N\}$ and $N = \text{card}\left(\mathcal{G}\right)$. We assume that the origin does not belong to $\mathcal{G}$. Thus, $\hat{\mathcal{R}}_\mathcal{G}$ can be written as follows:
\begin{align*}
	\hat{\mathcal{R}}_\mathcal{G} = \, \{x_i \in \mathcal{G} \,|\, \forall x_j\in \mathcal{G}, V_\theta(x_j) \le V_\theta(x_i) \Rightarrow \dot{V_\theta}(x_j) < 0\}.
\end{align*}

Intuitively, a point $x_i$ belongs to $\hat{\mathcal{R}}_\mathcal{G}$ if every other point with a smaller or equal Lyapunov value satisfies the strict decrease condition. This discrete definition mirrors the set $\hat{\mathcal{R}}$ introduced in Theorem \ref{our_theorem}.

Based on logical operators properties, we compute the cardinality of the $\hat{\mathcal{R}}_\mathcal{G}$, which gives the following expression: 
\begin{align*}
	&\text{card}(\hat{\mathcal{R}}_\mathcal{G}) = \text{card}(\mathcal{G}) - \nonumber \\ 
	&\, \sum_{i=1}^{N} \left( \min_{j=1:N} \left[ \min \left( \mathbb{I}(V_\theta(x_i) \ge V_\theta(x_j)),\mathbb{I}(\dot{V_\theta}(x_j) \ge 0) \right) \right] \right),
\end{align*}
where $\mathbb{I}(\cdot)$ denotes the boolean indicator function, i.e. $\mathbb{I}(\text{True}) = 1$ and $\mathbb{I}(\text{False}) = 0$.

Our aim is to maximize the size of $\hat{\mathcal{R}}_\mathcal{G}$ by adjusting the network parameters $\theta$. This leads to the loss function:
\begin{align}
	\label{eq:loss_function}
	\mathcal{L}(\theta) = -\text{card}\left( \hat{\mathcal{R}}_\mathcal{G} \right),
\end{align}
which encourages the Lyapunov candidate $V_\theta$ to satisfy the Lyapunov decrease condition on as many grid points as possible. The set $\hat{\mathcal{R}}_\mathcal{G}$ serves as a discrete inner approximation of the true region of attraction. The extension of this discrete approximation to a continuous-domain approximation of the RoA is discussed in subsection \ref{sec:verification}.

\begin{rem}
	The proposed loss function offers two main advantages. First, it consists of a single term per grid point, unlike existing approaches where the loss function includes multiple components and  penalize separately violations of Lyapunov conditions. This eliminates the need to tune weighting factors to balance competing objectives. Second,  our loss directly optimizes the size of the RoA estimate. Indeed, most neural-network-based RoA approximation methods follow a two-step strategy: identify a region where Lyapunov conditions hold, then compute the largest sublevel set within it. This strategy leads often to significantly smaller than the largest possible invariant set due to the conservative nature of the construction.
\end{rem}

\begin{figure*}[!t]
	\centering
	\begin{subfigure}[b]{0.32\textwidth}
		\centering
		\includegraphics[width=\textwidth]{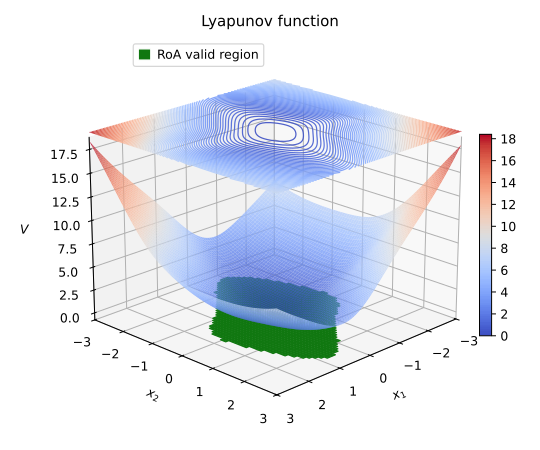}
		\caption{Lyapunov function}
		\label{fig:vanderpolsub1}
	\end{subfigure}
	\hfill
	\begin{subfigure}[b]{0.32\textwidth}
		\centering
		\includegraphics[width=\textwidth]{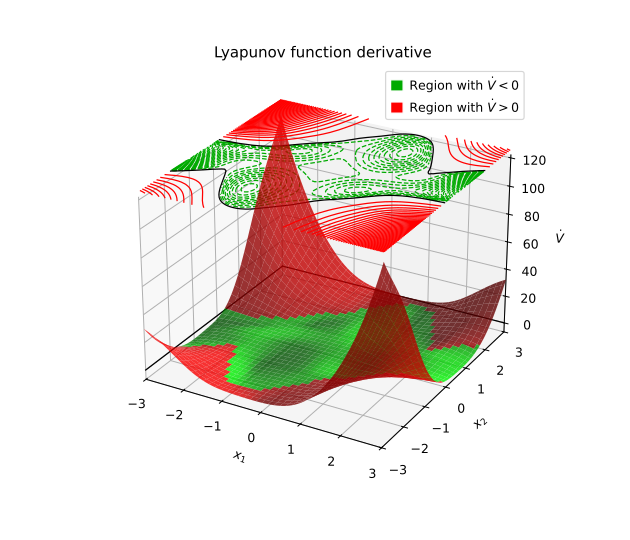}
		\caption{Lyapunov function derivative}
		\label{fig:vanderpolsub2}
	\end{subfigure}
	\hfill
	\begin{subfigure}[b]{0.31\textwidth}
		\centering
		\includegraphics[width=\textwidth]{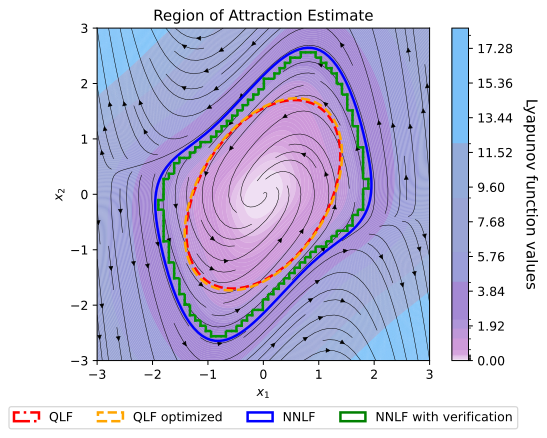} 
		\caption{Region of attraction estimate}
		\label{fig:vanderpolsub3}
	\end{subfigure}
	\caption{Van der Pol oscillator system}
	\label{fig:vanderpol}
\end{figure*}
\begin{figure*}[!t]
	\centering
	\begin{subfigure}[b]{0.32\textwidth}
		\centering
		\includegraphics[width=\textwidth]{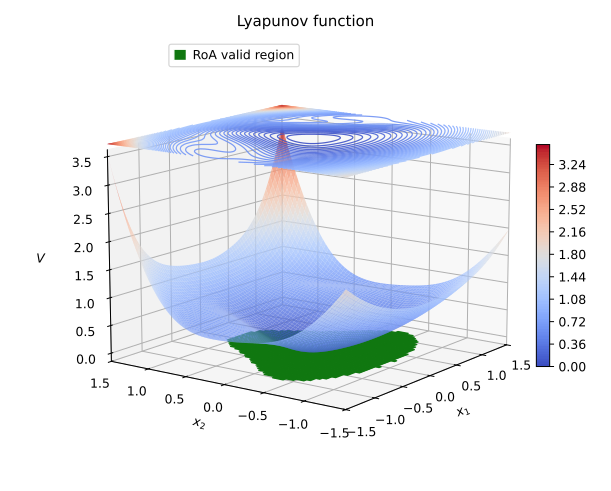}
		\caption{Lyapunov function}
		\label{fig:polynomialsub1}
	\end{subfigure}
	\hfill
	\begin{subfigure}[b]{0.32\textwidth}
		\centering
		\includegraphics[width=\textwidth]{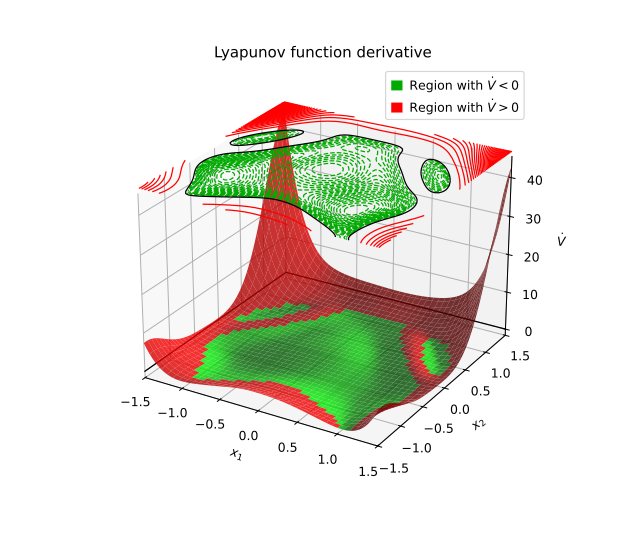}
		\caption{Lyapunov function derivative}
		\label{fig:polynomialsub2}
	\end{subfigure}
	\hfill
	\begin{subfigure}[b]{0.31\textwidth}
		\centering
		\includegraphics[width=\textwidth]{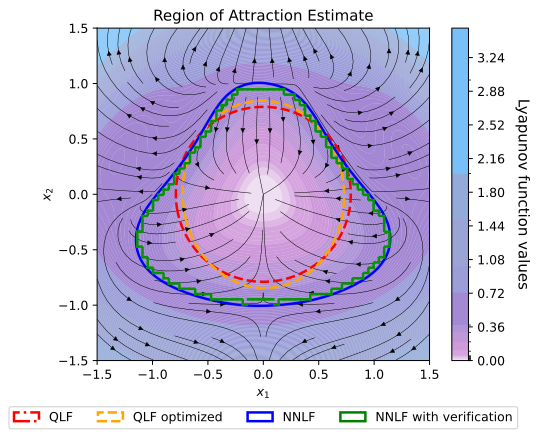} 
		\caption{Region of attraction estimate}
		\label{fig:polynomialsub3}
	\end{subfigure}
	\caption{Quartic Interaction system}
	\label{fig:polynomial}
\end{figure*}

\subsection{Neural network training for RoA learning}
The training step consists in adjusting the neural network parameters to minimize a loss function. Specifically, we employ the loss function defined in equation \eqref{eq:loss_function}, which is designed to guide the network toward maximizing the size of the discrete RoA estimate over the training set $\mathcal{G}$, composed of sampled states excluding the origin.

For implementation purposes, the indicator functions appearing in \eqref{eq:loss_function} are replaced by smooth approximations, such as weighted sigmoid functions, to enable gradient-based optimization. Furthermore, the strict decrease condition $\dot{V} < 0$ is reinforced to $\dot{V} \le -\varepsilon_{\dot{V}}$ with $\varepsilon_{\dot{V}} > 0$, introducing a margin that improves numerical robustness and prevents near-zero derivatives from falsely satisfying the condition.

\subsection{Verification of the Lyapunov function properties}
Once the learning phase is complete, we initiate a verification process to validate the learned candidate function $V_\theta$ over a continuous domain, ensuring that the learned function generalizes beyond the training data. Because $V_\theta$ and $\dot{V}_\theta$ take very small values near the origin, numerical artifacts can lead to incorrect sign conclusions. To address this issue, the verification is split into two stages: (i) analysis near the origin and (ii) verification around sampled training points.

\subsubsection{Stability analysis near the origin}
\label{sec:analysis_origin}
In this subsection, we establish a principled construction of the quadratic matrix $P_\theta$ that (i) preserves positive definiteness of $V_\theta$, and (ii) guarantees a strict decay of $\dot{V}_\theta$ in a neighborhood of the origin. To this end, we start by computing bounds on different involved mathematical models in $V$ and $\dot{V}$.

\paragraph{Smoothness bounds of the dynamics model}
We recall the system dynamics equation \eqref{eq:system_dynamics} and we assume that $f$ is twice continuously differentiable with $f(0)=0$. We introduce the Jacobian $J_f(x) \in \mathbb{R}^{n \times n}$ and the Hessian $H_f(x) \in \mathbb{R}^{n \times n \times n}$ of function $f$ at point $x \in \mathcal{D}$. 
We decompose $f(x)$ as follows:
\begin{align*}
	f(x) = A x + g(x), \text{ with } \frac{g(x)}{\| x \|} \to 0 \text{ as } \| x \| \to 0.
\end{align*}

By bounding the first and second-order Taylor remainder terms of the function $f$, one can establish the existence of finite constants
\begin{align*}
	C_f(x) = \sup_{t\in[0,1]}\|J_f(t x)\| \text{ and } C_g(x) = \frac{1}{2} \sup_{t\in[0,1]} \|H_f(t x)\|,
\end{align*}
such that, for all $x\in \mathcal{D}$, $\|f(x)\| \le C_f(x) \|x\|$ and $\|g(x)\| \le C_g(x) \|x\|^2$. Note that the constants $C_f(x)$ and $C_g(x)$ depend on $x$. In practice, we restrict the analysis to a sufficiently large ball $\mathcal{B}_\delta = \{ x\in\mathcal{D} \, | \, \|x\| \le \delta\}$ around the origin, with $\delta > 0$. We compute upper bounds $C_f = \sup_{x \in \mathcal{B}_\delta} C_f(x)$ and $C_g = \sup_{x \in \mathcal{B}_\delta} C_g(x)$ that are valid for all $x \in \mathcal{B}_\delta$. Thus,
\begin{align}
	\label{eq:Cf_and_Cg}
	\|f(x)\| \le C_f \|x\| \text{ and } \|g(x)\| \le C_g \|x\|^2, \quad \forall x \in \mathcal{B}_\delta.
\end{align}

\paragraph{Bounds of the neural network model}
Consider the one-hidden-layer neural network introduced in subsection \ref{subsec:design_LF}, and used to construct the Lyapunov candidate function $V_\theta$:
\begin{align*}
	\varphi_\theta(x) \;=\; W_2 \tanh(W_1 x + b_1),
\end{align*}
where $W_1\in\mathbb{R}^{N\times n}$, $W_2\in\mathbb{R}^{1\times N}$, $b_1\in\mathbb{R}^N$. A first-order expansion of $\varphi_\theta$ around $0$ yields
\begin{align*}
	\varphi_\theta(x)-\varphi_\theta(0) = W_2 D W_1 x + r(x);\, \frac{r(x)}{\| x \|} \to 0 \, \text{as} \, \| x \| \to 0,
\end{align*}
where $D = diag\left(1-\tanh^2(b_1)\right) \in \mathbb{R}^{N\times N}$, $r(x) \in \mathbb{R}$ is the second-order Taylor remainder term of $\varphi_\theta(x)$ around $0$, and
\begin{align*}
\frac{\partial \varphi_\theta}{\partial x}(x) = W_2 D W_1 + s(x);\; \|s(x)\| \to 0 \text{ as } \|x\| \to 0,
\end{align*}
where $s(x) \in \mathbb{R}^{1 \times n}$ is the first-order Taylor remainder term of $\frac{\partial \varphi_\theta}{\partial x}(x)$ around $0$.

Moreover, by the Taylor's remainder theorem and the $\tanh$ function properties, we get the following inequalities:
\begin{align}
	\label{eq:bound_r}
	\|r(x)\| &\le \|W_2\| \|W_1\|^2 \|x\|^2, \\ 
	\label{eq:bound_s}
	\|s(x)\| &\le \|W_2\|\|W_1\|^2 \|x\|.
\end{align} 

\paragraph{Design of matrix $P_\theta$}
For the construction of $P_\theta$, we start by computing $\dot{V}_\theta$. Based on the previous analysis,
\begin{align*}
	\dot{V}_\theta(x) &= 2 \left( \left( \varphi_\theta(x) -  \varphi_\theta(x) \right)^T \frac{\partial \varphi_\theta(x)}{\partial x} + x^T P_\theta \right) f(x) \\
	&= 2x^T\left( W_1^TD W_2^TW_2 D W_1 + P_\theta \right)Ax + \rho(x),
\end{align*}
with $\frac{\rho(x)}{\|x\|^2} \to 0$ as  $\|x\| \to 0$. 

Since $A$ is Hurwitz, there exists a positive definite matrix $P_1$ solving the Lyapunov equation $A^T P_1 + P_1 A = -I$. Thus, by setting $P_\theta = \alpha P_1 - W_1^TD W_2^TW_2 D W_1$ with $\alpha >0$, we get $\dot{V}(x) = -\alpha \|x\|^2 + \rho(x)$. This infers the existence of a neighborhood around the origin we can guarantee strict decay of $V_\theta$. Lastly, a sufficient condition for ensuring the positive definiteness of $P_\theta$ is to set $\alpha > \underline{\alpha}$, with $\underline{\alpha}  = \frac{\| W_1 \|^2 \| W_2 \|^2 }{\lambda_{\min}(P_1)}$. in the sequel, we assume that $\alpha = \beta \underline{\alpha}$, with $\beta > 1$.

\paragraph{Strict decay in the origin neighborhood}
The last step of our analysis is to derive a neighborhood around the origin where it is guaranteed that the proposed Lyapunov candidate function $V_\theta$ is strictly decreasing. This property is crucial to avoid numerical artifacts near the origin when performing the verification \ref{sec:verification}. Based of the proposed construction of $P_\theta$, a sufficient condition for ensuring strict negativity of $\dot{V}_\theta$ is $\|\rho(x)\| < \alpha \|x\|^2$. Implementing inequalities \eqref{eq:Cf_and_Cg}, \eqref{eq:bound_r} and \eqref{eq:bound_s} in $\|\rho(x)\|$ allow to derive the constant
\begin{align*}
	\varepsilon_\delta = \frac{\beta}{2\left( \lambda_{\min}(P_1) (2+\|W_1\| \delta) \|W_1\| C_f + \lambda_{\max}(P_1) \beta C_g \right)}.
\end{align*}

Since inequality \eqref{eq:Cf_and_Cg} is valid only in $\mathcal{B}_\delta$, the neighborhood around the origin guaranteeing strict decay of $V_\theta$ is the ball $\mathcal{B}_{\varepsilon} =\{x \in \mathcal{D}|\|x\| \le \varepsilon \}$, with $\varepsilon = \min(\delta,\varepsilon_\delta)$.

\subsubsection{Region of Attraction Verification}
\label{sec:verification}

The verification step aims to extend the discrete approximation RoA, denoted by $\hat{\mathcal{R}}_\mathcal{G}$, to a continuous domain. To achieve this, we construct a candidate RoA estimate as the union of hypercubes centered at the points $x_i \in \hat{\mathcal{R}}_\mathcal{G}$. Each hypercube is defined using the $\ell_\infty$ norm as $\mathcal{C}_\gamma(x_i) = \{ y \in \mathbb{R}^n|\| y - x_i \|_\infty \le \tfrac{\gamma}{2} \}$,
where $x_i \in \hat{\mathcal{R}}_\mathcal{G}$ and $\gamma$ is the side length. The constant $\gamma$ is selected so that the union $\bigcup_{x_i \in \hat{\mathcal{R}}_\mathcal{G}} \mathcal{C}_\gamma(x_i)$ forms a connected set containing the origin. In our configuration, $\gamma$ is equal to the grid step.

To check whether Lyapunov conditions are violated within this region, we formulate a nonlinear constraint satisfaction problem expressed by the first-order logic formula:
\begin{align}
	\label{eq:violation_formula}
	\Phi(x) = \big( \| x \| \ge \varepsilon \big) \land \big( \dot{V}_\theta(x) \ge 0 \big),
\end{align}
where $\varepsilon$ is a constant introduced in subsection \ref{sec:analysis_origin}. 

Solving \eqref{eq:violation_formula} identifies points, in the continuous domain $\bigcup_{x_i \in \hat{\mathcal{R}}_\mathcal{G}} \mathcal{C}_\gamma(x_i)$,
that violate the Lyapunov conditions. These points are referred to as counterexamples. The term $\left( \|x\| \ge \varepsilon \right)$ excludes the ball of radius $\varepsilon$ from the Lyapunov negation conditions, since it is proven that $\dot{V}_\theta(x) < 0$ holds fold for any $x \in \mathcal{B}_\varepsilon$. This prevents false counterexamples caused by numerical artifacts near the origin.

For implementation, we employ an SMT solver, specifically dReal. The problem \eqref{eq:violation_formula} is solved in parallel for all hypercubes $\mathcal{C}_\gamma(x_i)$ with $x_i \in \hat{\mathcal{R}}_\mathcal{G}$. If the solver cannot satisfy the formula within a given hypercube, that region is considered verified. Otherwise, the corresponding grid point $x_i$ and all points having higher Lyapunov levels are removed from $\hat{\mathcal{R}}_\mathcal{G}$ to match the conditions of Corollary \ref{cor:RoA_estimate}. In our experiments, we observe that the identified counterexamples lay on the border of the RoA estimate. This aspect will be the subject of future studies.

\section{Numerical examples}
We demonstrate the effectiveness of the proposed method for learning RoA estimates and neural-network-based Lyapunov functions across two nonlinear dynamical systems. In all experiments, we employ a one-hidden-layer neural network of $512$ neurons to learn the Lyapunov function, using the $\tanh$ activation function. The proposed loss function is minimized using the Adam optimizer with a learning rate of $5\times10^{-4}$. The training data is generated over a specified box domain by sampling points from a uniformly spaced $50 \times 50$ grid. For SMT-based verification, we use the dReal solver with a $\delta$ parameter equal to $10^{-6}$ times the grid step, and we set $\varepsilon = 10^{-4}$ to prevent numerical artifacts. 

To further validate the obtained RoA estimates, we compare them with a standard quadratic Lyapunov function $V(x) = x^T P_1 x$, where $P_1$ is the solution of the Lyapunov equation $A^TP_1 + P_1A = -Q$ with $Q=I$, and an optimized quadratic Lyapunov function $V(x) = x^T P x$, where $P$ is learned by minimizing our RoA-based loss function \eqref{eq:loss_function} over the space of positive definite matrices $Q$.

\subsubsection{Van der Pol oscillator}
The first system to which we apply our method is the Van der Pol oscillator system in reverse time, given by
\begin{align*}
	\begin{cases}
		\dot{x}_1 &= -x_2 \\
		\dot{x}_2 &= x_1 + (x_1^2-1)x_2
	\end{cases}.
\end{align*}

The above system has one asymptotically stable equilibrium point at $(0,0)$. The goal is to learn a Lyapunov function over the domain $[-3,3]\times[-3,3]$ that provides the largest possible RoA estimate. The results of training the proposed neural Lyapunov function are shown in Figure \ref{fig:vanderpol}. Specifically, Figure \ref{fig:vanderpolsub1} illustrates the shape of the learned Lyapunov function and the RoA estimate verified by the SMT solver. Figure \ref{fig:vanderpolsub2} shows the derivative of the neural Lyapunov function, while Figure \ref{fig:vanderpolsub3} shows the phase portrait and compares different RoA estimates: the trained neural network Lyapunov function (blue) and its verified inner approximation (green), the standard quadratic Lyapunov function (red), and the optimized quadratic Lyapunov function (orange). Note that for the two examples, the QLF and the optimized QLF RoA estimates are close because of the structure of the Jacobian matrices. From this comparison, we see that the trained neural network achieves the largest RoA estimate, which closely approximates the true RoA of the system.

\subsubsection{Quartic Interaction System}

The second system to which we apply our method is the following quartic Interaction System 
\begin{align*}
	\begin{cases}
		\dot{x}_1 &= -x_1(1-x_1^2-2x_2+x_2^4) \\
		\dot{x}_2 &= -x_2(1-x_1^4-x_2^2)
	\end{cases}.
\end{align*}

The above system has one asymptotically stable equilibrium point  at $(0,0)$ and 3 unstable equilibrium points at $(0,1)$, $(1,0)$, and $(−1,0)$. The goal is to learn a Lyapunov function over the given domain $[-1.5,1.5]\times[-1.5,1.5]$ that provides the largest possible RoA estimate. Results of learning and verifying the neural Lyapunov function are shown in Figure \ref{fig:polynomial}. 

\section{Conclusion}
We presented a new framework for learning and verifying neural-network-based Lyapunov functions to estimate RoA for nonlinear systems with formal guarantees. Our training approach significantly reduces the conservatism inherent in standard analytical methods, while ensuring formal stability guarantees via SMT-based verification. Experimental results demonstrate that it outperforms existing techniques in enlarging the verified RoA. This framework offers practical benefits for safety-critical applications such as robotics and autonomous systems, where certified stability is essential. Future work will focus on scaling to high-dimensional systems, integrating adaptive architectures for real-time learning, and developing efficient verification pipelines for online deployment. By bridging data-driven approximation with formal certification, this work sets a new direction for scalable and provably safe learning-based control in complex nonlinear environments.

\end{document}